\documentclass{llncs}

\usepackage{url}

\usepackage{amsmath}
\usepackage{amssymb}

\newcommand{\abs}[1]{\left\lvert#1\right\rvert}

\DeclareMathOperator{\Dom}{dom}
\newcommand{\rest}[2]{#1\!\!\restriction_{#2}}

\newcommand{\N}{\mathbb{N}}
\newcommand{\Q}{\mathbb{Q}}
\newcommand{\R}{\mathbb{R}}
\newcommand{\Rc}{\mathbb{R}_c}
\newcommand{\X}{\{0,1\}^*}
\newcommand{\K}{H}
\newcommand{\lad}{\underline{\dim}_A}
\newcommand{\lads}[1]{\underline{\dim}_A\{#1\}}

\newcommand{\noi}{\noindent}

\begin{document}


\begin{center}
{\Large \textbf{
Partial Randomness and Dimension of\\
Recursively Enumerable Reals
}}
\end{center}

\vspace{-3mm}

\begin{center}
Kohtaro Tadaki
\end{center}

\vspace{-5mm}

\begin{center}
Research and Development Initiative, Chuo University\\
JST, CREST\\
1--13--27 Kasuga, Bunkyo-ku, Tokyo 112-8551, Japan\\
E-mail: tadaki@kc.chuo-u.ac.jp
\end{center}

\vspace{-2mm}

\begin{quotation}
\noi\textbf{Abstract.}
A real $\alpha$ is called recursively enumerable
(``r.e.'' for short)
if
there exists a computable,
increasing sequence of rationals which converges to $\alpha$.
It is known that
the randomness of an r.e.~real
$\alpha$
can be characterized in various ways
using each of the notions;
program-size complexity,
Martin-L\"{o}f test,
Chaitin $\Omega$ number,
the domination and $\Omega$-likeness of $\alpha$,
the universality of
a computable, increasing sequence of rationals
which converges to $\alpha$,
and universal probability.
In this paper,
we generalize these characterizations of randomness
over the notion of partial randomness
by parameterizing each of the notions above
by a real $T\in(0,1]$,
where the notion of partial randomness is
a stronger representation of
the compression rate
by means of
program-size complexity.
As a result,
we present ten equivalent characterizations of
the partial randomness of an r.e.~real.
The resultant characterizations
of partial randomness
are powerful and
have many important applications.
One of them is
to present equivalent characterizations of
the dimension of an individual r.e.~real.
The equivalence between
the notion of 
Hausdorff dimension and
compression rate by program-size complexity
(or partial randomness)
has been established at present
by a series of works of many researchers
over the last two decades.
We
present ten equivalent characterizations of
the dimension of an
individual
r.e.~real.
\end{quotation}

\begin{quotation}
\noi\textit{Key words\/}:
algorithmic randomness,
recursively enumerable real,
partial randomness,
dimension,
Chaitin $\Omega$ number,
program-size complexity,
universal probability
\end{quotation}

\section{Introduction}

A real $\alpha$ is called \textit{recursively enumerable}
(``r.e.'' for short)
if there exists a computable, increasing sequence of rationals
which converges to $\alpha$.
The randomness of an r.e.~real $\alpha$
can be characterized in various ways
using each of the notions;
\textit{program-size complexity},
\textit{Martin-L\"{o}f test},
\textit{Chaitin $\Omega$ number},
the \textit{domination} and \textit{$\Omega$-likeness}
of $\alpha$,
the \textit{universality} of
a computable, increasing sequence of rationals
which converges to $\alpha$,
and \textit{universal probability}.
These equivalent characterizations of randomness of an r.e.~real
are summarized in Theorem~\ref{randomness}
(see Section~\ref{previous results}),
where the equivalences are
established by a series of works of
Martin-L\"{o}f \cite{M66},
Schnorr \cite{Sch73},
Chaitin \cite{C75},
Solovay \cite{Sol75},
Calude, Hertling, Khoussainov and Wang \cite{CHKW01},
Ku\v{c}era and Slaman \cite{KS01},
and Tadaki \cite{T06},
between 1966 and 2006.
In this paper,
we generalize
these characterizations of randomness
over the notion of \textit{partial randomness},
which was introduced by Tadaki \cite{T99,T02}
and is a stronger representation of \textit{the compression rate}
by means of program-size complexity.
We introduce
many
characterizations of
partial randomness for an r.e.~real
by parameterizing each of the notions above on randomness
by a real $T\in(0,1]$.
In particular,
we introduce the notion of \textit{$T$-convergence}
for a computable, increasing sequence of rationals
and then introduce the same notion for an r.e.~real.
The notion of $T$-convergence
plays a crucial role in
these our
characterizations of
partial randomness.
We then prove the equivalence of all
these
characterizations
of partial randomness in Theorem \ref{partial randomness},
our main result,
in Section \ref{Our results}.

On the other hand,
by a series of works of
Ryabko \cite{Ryab84,Ryab86},
Staiger \cite{St93,St98},
Tadaki \cite{T99,T02},
Lutz \cite{Lutz00}, and
Mayordomo \cite{Mayor02}
over the last two decades,
the equivalence between
the notion of compression rate by program-size complexity (or partial randomness)
and Hausdorff dimension seems to be established at present.
The subject of the equivalence seems to be
one of the most active areas of the recent research
of algorithmic randomness.
In the context of
the
subject,
we can consider the notion of the \textit{dimension} of an individual real
in particular,
and this notion plays a crucial role in the subject.
As
one of the main applications of our main result on partial randomness,
i.e., Theorem~\ref{partial randomness},
we
can
present many equivalent characterizations of the dimension of
an
individual
r.e.~real.

The paper is organized as follows.
We begin in Section~\ref{preliminaries} with
some preliminaries to
algorithmic information theory and partial randomness.
In Section~\ref{previous results},
we review the previous results on
the equivalent characterizations of randomness of an r.e.~real.
Our main result on partial randomness of an r.e.~real is
presented in Section~\ref{Our results}.
In Section~\ref{dimension}
we apply our main result on partial randomness
to give many equivalent characterizations of the dimension of an r.e.~real.
In Section~\ref{T-convergence},
we investigate further properties of the notion of $T$-convergence,
which plays a crucial role in our characterizations of
the partial randomness and dimension of r.e.~reals.
We conclude this paper with
a mention of
the future direction of this work in Section~\ref{conclusion}.
Due to the 12-page limit, we omit most proofs.
A full paper which describes all the proofs and other related results
is in preparation.

\vspace*{-3mm}

\section{Preliminaries}
\label{preliminaries}

\vspace*{-3mm}


We start with some notation about numbers and strings
which will be used in this paper.
$\N=\left\{0,1,2,3,\dotsc\right\}$ is the set of natural numbers,
and $\N^+$ is the set of positive integers.
$\Q$ is the set of rational numbers, and
$\R$ is the set of real numbers.
A sequence $\{a_n\}_{n\in\N}$ of numbers
(rationals or reals)
is called \textit{increasing} if $a_{n+1}>a_{n}$ for all $n\in\N$.

$\X=
\left\{
  \lambda,0,1,00,01,10,11,000,
  \dotsc
\right\}$
is the set of finite binary strings,
where $\lambda$ denotes the \textit{empty string}.
For any $s \in \X$, $\abs{s}$ is the \textit{length} of $s$.
A subset $S$ of $\X$ is called \textit{prefix-free}
if no string in $S$ is a prefix of another string in $S$.
For any partial function $f$,
the domain of definition of $f$ is denoted by $\Dom f$.
We write ``r.e.'' instead of ``recursively enumerable.''

Normally, $o(n)$ denotes any
function $f\colon \N^+\to\R$ such
that $\lim_{n \to \infty}f(n)/n=0$.
On the other hand,
$O(1)$ denotes any
function $g\colon \N^+\to\R$ such that
there is $C\in\R$ with the property that
$\abs{g(n)}\le C$ for all $n\in\N^+$.

Let $\alpha$ be an arbitrary real.
For any $n\in\N^+$,
we denote by $\rest{\alpha}{n}\in\X$
the first $n$ bits of the base-two expansion of
$\alpha - \lfloor \alpha \rfloor$ with infinitely many zeros,
where $\lfloor \alpha \rfloor$ is the greatest integer
less than or equal to $\alpha$.
Thus,
in particular,
if $\alpha\in[0,1)$,
then $\rest{\alpha}{n}$ denotes the first $n$ bits of
the base-two expansion of $\alpha$ with infinitely many zeros.
For example,
in the case of $\alpha=5/8$,
$\rest{\alpha}{6}=101000$.

A real $\alpha$ is called \textit{r.e.}~if
there exists a computable,
increasing sequence of rationals which converges to $\alpha$.
An r.e.~real is also called a
\textit{left-computable} real.
Let $\alpha$ and $\beta$ be arbitrary r.e.~reals.
Then $\alpha+\beta$ is r.e.
If $\alpha$ and $\beta$ are non-negative, then $\alpha\beta$ is r.e.
On the other hand,
a real $\alpha$ is called \textit{right-computable} if
$-\alpha$ is left-computable.
We say that a real $\alpha$ is \textit{computable} if
there exists a computable sequence $\{a_n\}_{n\in\N}$ of rationals
such that $\abs{\alpha-a_n} < 2^{-n}$ for all $n\in\N$.
It is then easy to see that,
for every $\alpha\in\R$,
$\alpha$ is computable if and only if
$\alpha$ is both left-computable and right-computable.
A sequence $\{a_n\}_{n\in\N}$ of reals is called
\textit{computable} if
there exists a total recursive function $f\colon\N\times\N \to \Q$
such that
$\abs{a_n-f(n,m)} < 2^{-m}$ for all $n, m\in\N$.
See e.g.~Weihrauch \cite{W00}
for the detail of the treatment of
the computability of reals and
sequences of reals.

\vspace*{-2mm}

\subsection{Algorithmic information theory}
\label{ait}

In the following
we concisely review some definitions and results of
algorithmic information theory
\cite{C75,C87b}.
A \textit{computer} is a partial recursive function
$C\colon \X\to \X$
such that
$\Dom C$ is a prefix-free set.
For each computer $C$ and each $s \in \X$,
$H_C(s)$ is defined as
$\min
\left\{\,
  \abs{p}\,\big|\;p \in \X\>\&\>C(p)=s
\,\right\}$
(may be $\infty$).
A computer $U$ is said to be \textit{optimal} if
for each computer $C$ there exists $d\in\N$
with the following property;
if $p\in\Dom C$, then there is $q\in\Dom U$ for which
$U(q)=C(p)$ and $\abs{q}\le\abs{p}+d$.
It is easy to see that there exists an optimal computer.
We choose a particular optimal computer $U$
as the standard one for use,
and define $H(s)$ as $H_U(s)$,
which is referred to as
the \textit{program-size complexity} of $s$ or
the \textit{Kolmogorov complexity} of $s$.
It follows that
for every computer $C$ there exists $d\in\N$ such that,
for every $s\in\X$, $H(s)\le H_C(s)+d$.

For any optimal computer $V$,
\textit{Chaitin's halting probability} $\Omega_V$ of $V$ is defined
as $\sum_{p\in\Dom V}2^{-\abs{p}}$.
The real $\Omega_V$ is also called \textit{Chaitin $\Omega$ number}.


\begin{definition}[weak Chaitin randomness, Chaitin \cite{C75,C87b}]
  For any $\alpha\in\R$,
  we say that $\alpha$ is \textit{weakly Chaitin random} if
  there exists $c\in\N$ such that
  $n-c \le H(\rest{\alpha}{n})$
  for all $n\in\N^+$.
  \qed
\end{definition}

Chaitin \cite{C75} showed that,
for every optimal computer $V$,
$\Omega_V$ is weakly Chaitin random.

\begin{definition}[Martin-L\"{o}f randomness, Martin-L\"{o}f \cite{M66}]
A subset $\mathcal{C}$ of $\N^+\times\X$
is called a \textit{Martin-L\"{o}f test} if
$\mathcal{C}$ is an r.e.~set and
\begin{equation*}
  \forall\,n\in\N^+\quad
  \sum_{s\,\in\,\mathcal{C}_n}2^{-\abs{s}}
  \le 2^{-n},
\end{equation*}
where
$\mathcal{C}_n
=
\left\{\,
  s\bigm|(n,s)\in\mathcal{C}
\,\right\}$.
For any $\alpha\in\R$,
we say that $\alpha$ is \textit{Martin-L\"{o}f random} if
for every Martin-L\"{o}f test $\mathcal{C}$,
there exists $n\in\N^+$ such that, for every $k\in\N^+$,
$\rest{\alpha}{k}\notin\mathcal{C}_n$.\qed
\end{definition}

\begin{theorem}[Schnorr \cite{Sch73}]\label{wcem}
  For every $\alpha\in\R$,
  $\alpha$ is weakly Chaitin random if and only if
  $\alpha$ is Martin-L\"{o}f random.\qed
\end{theorem}

The program-size complexity $\K(s)$ is originally defined
using the concept of program-size, as stated above.
However,
it is possible to define $\K(s)$ without referring to such a concept,
i.e.,
as in the following,
we first introduce a \textit{universal probability} $m$,
and then define $\K(s)$ as $-\log_2 m(s)$.
A universal probability is defined as follows \cite{ZL70}.

\begin{definition}[universal probability]
  A function $r\colon \X\to[0,1]$ is called
  a \textit{lower-computable semi-measure} if
  $\sum_{s\in \X}r(s)\le 1$ and
  the set $\{(a,s)\in\Q\times\X\mid a<r(s)\}$ is r.e.
  We say that a lower-computable semi-measure $m$ is
  a \textit{universal probability} if
  for every lower-computable semi-measure $r$,
  there exists $c\in\N^+$ such that,
  for all $s\in \X$, $r(s)\le cm(s)$.\qed
\end{definition}

The following theorem can be then shown
(see e.g.~Theorem 3.4 of Chaitin \cite{C75} for its proof).

\begin{theorem}%
\label{eup}
  For every optimal computer $V$,
  the function $2^{-\K_V(s)}$ of $s$ is a universal probability.
  \qed
\end{theorem}

By Theorem~\ref{eup},
we see that $\K(s)=-\log_2 m(s)+O(1)$
for every universal probability $m$.
Thus it is possible to define $\K(s)$ as $-\log_2 m(s)$
with a particular universal probability $m$
instead of as $\K_U(s)$.
Note that
the difference up to an additive constant is
nonessential
to algorithmic information theory.

\vspace*{-2mm}

\subsection{Partial randomness}
\label{partial}

In the works \cite{T99,T02},
we generalized the notion of
the randomness of
a real
so that \textit{the degree of the randomness},
which is often referred to
as
\textit{the partial randomness} recently
\cite{CST06,RS05,CS06},
can be characterized by a real $T$
with $0<T\le 1$ as follows.

\begin{definition}[\boldmath weak Chaitin $T$-randomness]
  Let $T\in\R$ with $T\ge 0$.
  For any $\alpha\in\R$,
  we say that $\alpha$ is \textit{weakly Chaitin $T$-random} if
  there exists $c\in\N$ such that
  $Tn-c \le H(\rest{\alpha}{n})$
  for all $n\in\N^+$.
  \qed
\end{definition}

\begin{definition}[\boldmath Martin-L\"{o}f $T$-randomness]
Let $T\in\R$ with $T\ge 0$.
A subset $\mathcal{C}$ of $\N^+\times\X$
is called a \textit{Martin-L\"{o}f $T$-test} if
$\mathcal{C}$ is an r.e.~set and
\begin{equation*}
  \forall\,n\in\N^+\quad
  \sum_{s\,\in\,\mathcal{C}_n}2^{-T\abs{s}}
  \le 2^{-n}.
\end{equation*}
For any $\alpha\in\R$,
we say that $\alpha$ is \textit{Martin-L\"{o}f $T$-random} if
for every Martin-L\"{o}f $T$-test $\mathcal{C}$,
there exists $n\in\N^+$ such that, for every $k\in\N^+$,
$\rest{\alpha}{k}\notin\mathcal{C}_n$.\qed
\end{definition}

In the case where $T=1$,
the weak Chaitin $T$-randomness and Martin-L\"{o}f $T$-randomness
result in weak Chaitin randomness and Martin-L\"{o}f randomness,
respectively.
Tadaki \cite{T02} generalized Theorem~\ref{wcem}
over the notion of $T$-randomness
as follows.

\begin{theorem}[Tadaki \cite{T02}]\label{Twcem}
  Let $T$ be a computable real with $T\ge 0$.
  Then,
  for every $\alpha\in\R$,
  $\alpha$ is weakly Chaitin $T$-random if and only if
  $\alpha$ is Martin-L\"{o}f $T$-random.\qed
\end{theorem}

\begin{definition}[\boldmath $T$-compressibility]
Let $T\in\R$ with $T\ge 0$.
For any $\alpha\in\R$,
we say that $\alpha$ is \textit{$T$-compressible} if
$H(\rest{\alpha}{n})\le Tn+o(n)$,
which is equivalent to
$\limsup_{n \to \infty}H(\rest{\alpha}{n})/n\le T$.
\qed
\end{definition}

For every $T\in[0,1]$ and every $\alpha\in\R$,
if $\alpha$ is weakly Chaitin $T$-random and $T$-compressible,
then
\begin{equation}\label{compression-rate}
  \lim_{n\to \infty} \frac{H(\rest{\alpha}{n})}{n} = T.
\end{equation}
The
left-hand side of \eqref{compression-rate}
is referred to as the \textit{compression rate} of
a real $\alpha$ in general.
Note, however, that \eqref{compression-rate}
does not necessarily imply that $\alpha$ is weakly Chaitin $T$-random.
Thus, the notion of partial randomness
is
a stronger representation of compression rate.

In the works \cite{T99,T02},
we generalized
Chaitin $\Omega$ number
to $\Omega(T)$
as follows.
For each optimal computer $V$ and each real $T>0$,
the \textit{generalized halting probability} $\Omega_V(T)$ of $V$ is
defined
by
\vspace*{-2mm}
\begin{equation*}
  \Omega_V(T) = \sum_{p\in\Dom V}2^{-\frac{\abs{p}}{T}}.
\end{equation*}
Thus,
$\Omega_V(1)=\Omega_V$.
If $0<T\le 1$, then $\Omega_V(T)$ converges and $0<\Omega_V(T)<1$,
since $\Omega_V(T)\le \Omega_V<1$.
The following theorem holds for $\Omega_V(T)$.

\begin{theorem}[Tadaki \cite{T99,T02}]\label{pomgd}
Let $V$ be an optimal computer and let $T\in\R$.
\begin{enumerate}
  \item If $0<T\le 1$ and $T$ is computable,
    then $\Omega_V(T)$ is weakly Chaitin $T$-random and
    $T$-compressible.
  \item If $1<T$, then $\Omega_V(T)$ diverges to $\infty$.\qed
\end{enumerate}
\end{theorem}

\vspace*{-3mm}

\section{Previous results on the randomness of an r.e.~real}
\label{previous results}

In this section,
we review the previous results on the randomness of an r.e.~real.
First we review some notions on r.e.~reals.

\begin{definition}[\boldmath $\Omega$-likeness]
For any r.e.~reals $\alpha$ and $\beta$,
we say that $\alpha$ \textit{dominates} $\beta$
if there are computable, increasing sequences $\{a_n\}$ and $\{b_n\}$
of rationals and $c\in\N^+$
such that
$\lim_{n\to\infty} a_n =\alpha$,
$\lim_{n\to\infty} b_n =\beta$,
and $c(\alpha-a_n)\ge \beta-b_n$ for all $n\in\N$.
An r.e.~real $\alpha$ is called \textit{$\Omega$-like}
if it dominates all r.e.~reals.
\qed
\end{definition}

Solovay
\cite{Sol75}
showed the following theorem.
For
its
proof, see also Theorem 4.9 of \cite{CHKW01}.

\begin{theorem}[Solovay \cite{Sol75}]\label{Solovay}
For every r.e.~reals $\alpha$ and $\beta$,
if $\alpha$ dominates $\beta$ then
$H(\rest{\beta}{n})\le H(\rest{\alpha}{n})+O(1)$
for all $n\in\N^+$.
\qed
\end{theorem}

\begin{definition}[universality]
A computable, increasing and converging sequence $\{a_n\}$
of rationals is called \textit{universal}
if for every computable, increasing and converging sequence $\{b_n\}$
of rationals
there exists $c\in\N^+$ such that
$c(\alpha-a_n)\ge \beta-b_n$ for all $n\in\N$,
where $\alpha=\lim_{n\to\infty} a_n$ and $\beta=\lim_{n\to\infty} b_n$.
\qed
\end{definition}

The previous results on
the equivalent characterizations of randomness for an r.e.~real
are summarized in the following theorem.

\begin{theorem}[\cite{Sch73,C75,Sol75,CHKW01,KS01,T06}]\label{randomness}
Let $\alpha$ be an r.e.~real with $0<\alpha<1$.
Then the following conditions are equivalent:
\begin{enumerate}
\item The real $\alpha$ is weakly Chaitin random.
\item The real $\alpha$ is Martin-L\"{o}f random.
\item The real $\alpha$ is $\Omega$-like.
\item For every r.e.~real $\beta$,
  $H(\rest{\beta}{n})\le H(\rest{\alpha}{n})+O(1)$
  for all $n\in\N^+$.
\item There exists an optimal computer $V$ such that $\alpha=\Omega_V$.
\item There exists a universal probability $m$ such that
  $\alpha=\sum_{s\in\X}m(s)$.
\item Every computable, increasing sequence of rationals
  which converges to $\alpha$ is universal.
\item There exists a universal computable, increasing sequence of
  rationals which converges to $\alpha$.\qed
\end{enumerate}
\end{theorem}

The historical remark on the proofs of equivalences
in Theorem \ref{randomness} is as follows.
Schnorr \cite{Sch73} showed that
(i) and (ii) are equivalent to each other.
Chaitin \cite{C75} showed that (v) implies (i).
Solovay \cite{Sol75} showed that
(v) implies (iii), (iii) implies (iv), and (iii) implies (i).
Calude, Hertling, Khoussainov, and Wang \cite{CHKW01} showed that
(iii) implies (v), and (v) implies (vii).
Ku\v{c}era and Slaman \cite{KS01} showed that
(ii) implies (vii).
Finally,
(vi) was inserted in the course of the derivation from (v) to (viii)
by Tadaki \cite{T06}.

\section{New results on the partial randomness of an r.e.~real}
\label{Our results}

\vspace*{-1mm}

In this section,
we generalize
Theorem \ref{randomness} above
over the notion of partial randomness.
For that purpose,
we first introduce some new notions.
Let $T$ be an arbitrary
real with $0<T\le 1$
throughout
the rest of this paper.
These notions are parametrized by the real $T$.%
\footnote{
The parameter $T$ corresponds to the notion of ``temperature''
in the statistical mechanical interpretation of
algorithmic information theory
developed by Tadaki \cite{T08CiE,T09LFCS}.
}

\begin{definition}[\boldmath $T$-convergence]
An increasing sequence $\{a_n\}$ of reals is called
\textit{$T$-convergent} if
$\sum_{n=0}^{\infty} (a_{n+1}-a_{n})^T\!<\infty$.
An r.e.~real $\alpha$ is called \textit{$T$-convergent} if
there exists a $T$-convergent computable,
increasing sequence of rationals which
converges to $\alpha$,
i.e.,
if there exists an increasing sequence $\{a_n\}$ of rationals
such that
(i) $\{a_n\}$ is $T$-convergent,
(ii) $\{a_n\}$ is computable,
and (iii) $\lim_{n\to\infty} a_n=\alpha$.
\qed
\end{definition}

\vspace*{-2mm}

Note that
every increasing and converging sequence of reals is
$1$-convergent, and thus
every r.e.~real is $1$-convergent.
In general,
based on the following
lemma,
we can freely switch from
``$T$-convergent computable,
increasing sequence of reals''
to
``$T$-convergent computable,
increasing sequence of rationals.''

\begin{lemma}\label{tcrr}
For every $\alpha\in\R$,
$\alpha$ is an r.e.~$T$-convergent real
if and only if
there exists a $T$-convergent computable,
increasing sequence of reals which
converges to $\alpha$.
\qed
\end{lemma}

The following argument illustrates the way of using
Lemma
\ref{tcrr}:
Let $V$ be an optimal computer, and
let $p_0,p_1,p_2,\dots$ be a recursive enumeration of the r.e.~set $\Dom V$.
Then $\Omega_V(T)=\sum_{i=0}^{\infty} 2^{-\abs{p_i}/T}$,
and the increasing sequence
$\left\{\sum_{i=0}^{n} 2^{-\abs{p_i}/T} \right\}_{n\in\N}$
of reals is $T$-convergent
since
$\Omega_V=\sum_{i=0}^{\infty} 2^{-\abs{p_i}}<1$.
If $T$ is computable,
then this sequence of reals is
computable. 
Thus,
by Lemma \ref{tcrr}
we have Theorem \ref{tomegavt} below.

\begin{theorem}\label{tomegavt}
Let $V$ be an optimal computer.
If $T$ is computable,
then $\Omega_V(T)$ is an r.e.~$T$-convergent real.
\qed
\end{theorem}

\begin{definition}[\boldmath $\Omega(T)$-likeness]
An r.e.~real $\alpha$ is called \textit{$\Omega(T)$-like}
if it dominates all r.e.~$T$-convergent reals.
\qed
\end{definition}

Note that
an r.e.~real $\alpha$ is $\Omega(1)$-like
if and only if $\alpha$ is $\Omega$-like.

\begin{definition}[\boldmath $T$-universality]
A computable, increasing and converging sequence $\{a_n\}$
of rationals is called \textit{$T$-universal}
if for every $T$-convergent computable,
increasing and converging sequence $\{b_n\}$ of
rationals
there exists $c\in\N^+$ such that
$c(\alpha-a_n)\ge \beta-b_n$ for all $n\in\N$,
where $\alpha=\lim_{n\to\infty} a_n$ and $\beta=\lim_{n\to\infty} b_n$.
\qed
\end{definition}

Note that
a computable, increasing and converging sequence $\{a_n\}$
of rationals is $1$-universal
if and only if $\{a_n\}$ is universal.

Using the notions introduced above,
Theorem~\ref{randomness} is generalized as follows.

\begin{theorem}[main result]\label{partial randomness}
Let $\alpha$ be an r.e.~real with $0<\alpha<1$.
Suppose that $T$ is computable.
Then the following conditions are equivalent:
\begin{enumerate}
\item The real $\alpha$ is weakly Chaitin $T$-random.
\item The real $\alpha$ is Martin-L\"{o}f $T$-random.
\item The real $\alpha$ is $\Omega(T)$-like.
\item For every r.e.~$T$-convergent real $\beta$,
  $H(\rest{\beta}{n})\le H(\rest{\alpha}{n})+O(1)$
  for all $n\in\N^+$.
\item For every r.e.~$T$-convergent real $\gamma>0$,
  there exist an r.e.~real $\beta\ge 0$ and
  a rational $q>0$ such that
  $\alpha=\beta+q\gamma$.
\item For every optimal computer $V$,
  there exist an r.e.~real $\beta\ge 0$ and
  a rational $q>0$ such that
  $\alpha=\beta+q\Omega_V(T)$.
\item There exist
  an optimal computer $V$ and an r.e.~real $\beta\ge 0$
  such that $\alpha=\beta+\Omega_V(T)$.
\item There exists a universal probability $m$ such that
  $\alpha=\sum_{s\in\X}m(s)^{\frac{1}{T}}$.
\item Every computable, increasing sequence of rationals
  which converges to $\alpha$ is $T$-universal.
\item There exists a $T$-universal computable, increasing sequence
  of rationals which converges to $\alpha$.\qed
\end{enumerate}
\end{theorem}

We see that
Theorem~\ref{partial randomness} is
a
massive expansion
of Theorem~\ref{Twcem}
in the case where the real $\alpha$ is r.e.~with $0<\alpha<1$.
The condition (vii) of Theorem~\ref{partial randomness}
corresponds to the condition (v) of Theorem~\ref{randomness}.
Note, however, that,
in the condition (vii) of Theorem~\ref{partial randomness},
a non-negative r.e.~real $\beta$ is needed.
The reason is as follows:
In the case of $\beta=0$,
the possibility that $\alpha$ is weakly Chaitin $T'$-random
with a real $T'>T$
is excluded by
the $T$-compressibility of $\Omega_V(T)$
imposed by Theorem~\ref{pomgd} (i).
However, this exclusion is inconsistent with
the condition (i) of Theorem~\ref{partial randomness}.

Theorem~\ref{partial randomness} can be  proved
by generalizing the proof of Theorem~\ref{randomness}
over the notion of partial randomness.
For example,
using Lemma~\ref{ks} below,
the implication (ii) $\Rightarrow$ (v) of Theorem~\ref{partial randomness}
is
proved as follows,
in which the notion of $T$-convergence
plays an important role.

\begin{proof}[of (ii) $\Rightarrow$ (v) of Theorem~\ref{partial randomness}]
Suppose that
$\gamma$ is an arbitrary r.e.~$T$-convergent real
with $\gamma>0$.
Then there exists a $T$-convergent computable,
increasing sequence $\{c_n\}$ of rationals which
converges to $\gamma$.
Since $\gamma>0$,
without loss of generality
we can assume that $c_{0}=0$.
We choose any one rational $\varepsilon>0$
such that
$\sum_{n=0}^{\infty} [\varepsilon (c_{n+1}-c_{n})]^T\le 1$.
Such $\varepsilon$ exists
since the sequence $\{c_n\}$ is $T$-convergent.
Note that the sequence $\{\varepsilon (c_{n+1}-c_{n})\}$
is a computable sequence of positive rationals.
Thus,
since $\alpha$ is a positive r.e.~real and also Martin-L\"{o}f $T$-random
by the assumption,
it follows from Lemma \ref{ks} below
that there exist a computable, increasing sequence $\{a_n\}$
of rationals and a rational $r>0$ such that
$a_{n+1}-a_{n}>r\varepsilon (c_{n+1}-c_{n})$ for every $n\in\N$,
$a_{0}>0$, and
$\alpha=\lim_{n\to\infty}a_{n}$.
We then define a sequence $\{b_n\}$ of positive
rationals
by $b_n=a_{n+1}-a_{n}-r\varepsilon (c_{n+1}-c_{n})$.
It follows that $\{b_n\}$ is a computable sequence of
rationals and
$\sum_{n=0}^{\infty}b_n$ converges to
$\alpha-a_{0}-r\varepsilon (\gamma-c_{0})$.
Thus we have
$\alpha=a_{0}+\sum_{n=0}^{\infty}b_n+r\varepsilon \gamma$,
where $a_{0}+\sum_{n=0}^{\infty}b_n$ is a positive r.e.~real.
This completes the proof.
\qed
\end{proof}

\begin{lemma}\label{ks}
Let $\alpha$ be an r.e.~real,
and
let $\{d_n\}$ be a computable sequence
of positive rationals such that
$\sum_{n=0}^{\infty} {d_n}^T\le 1$.
If $\alpha$ is Martin-L\"{o}f $T$-random,
then for every $\varepsilon>0$
there exist a computable, increasing sequence $\{a_n\}$
of rationals and a rational $q>0$ such that
$a_{n+1}-a_{n}>qd_{n}$ for every $n\in\N$,
$a_{0}>\alpha-\varepsilon$, and
$\alpha=\lim_{n\to\infty}a_{n}$.
\qed
\end{lemma}

Lemma~\ref{ks} can be proved,
based on the generalization of the techniques used
in the proof of Theorem 2.1 of Ku\v{c}era and Slaman \cite{KS01}
over partial randomness.
In addition to the proof of Lemma~\ref{ks},
the complete proof of Theorem~\ref{partial randomness}
will be described in a full version of this paper,
which is in preparation.

Theorem~\ref{partial randomness} has many important applications.
One of the main applications is
to give many characterizations of the dimension of
an individual r.e.~real,
some of
which will be presented in the next section.
As another consequence of Theorem~\ref{partial randomness},
we can obtain Corollary~\ref{omegat-omegat} below
for example,
which follows immediately from
the implication (vii) $\Rightarrow$ (iv) of
Theorem~\ref{partial randomness}
and Theorem \ref{tomegavt}.

\begin{corollary}\label{omegat-omegat}
Suppose that $T$ is computable.
Then, for every two optimal computers $V$ and $W$,
$H(\rest{\Omega_{V}(T)}{n})=H(\rest{\Omega_{W}(T)}{n})+O(1)$
for all $n\in\N^+$.
\qed
\end{corollary}

Note that
the computability of $T$ is important
for Theorem~\ref{partial randomness} to hold.
For example,
we cannot allow $T$ to be simply an r.e.~real
in Theorem~\ref{partial randomness}.

The notion of $T$-convergence has many interesting properties,
in addition to the properties which we saw above.
In Section~\ref{T-convergence},
we investigate further properties of the notion of $T$-convergence.

\vspace*{-1mm}

\section{New characterizations of the dimension of an r.e.~real}
\label{dimension}

In this section we apply Theorem~\ref{partial randomness} to
give many characterizations of dimension for an individual r.e.~real.
In the works \cite{T99,T02},
we introduced
the notions of six ``algorithmic dimensions'',
1st, 2nd, 3rd, 4th, upper, and lower algorithmic
dimensions as fractal dimensions for a subset $F$ of
$N$-dimensional Euclidean space $\R^N$.
These notions are defined
based on the notion of partial randomness and compression rate
by means of program-size complexity.
We then showed that
all the six algorithmic dimensions equal to
the Hausdorff dimension for any self-similar set which
is computable in a certain sense.
The class of such self-similar sets includes familiar fractal sets
such as the Cantor set, von Koch curve, and Sierpi\'nski gasket.
In particular,
the notion of lower algorithmic dimension
for a subset $F$ of $\R$
is defined as follows.

\begin{definition}[lower algorithmic dimension, Tadaki \cite{T02}]
Let $F$ be a nonempty subset of $\R$.
The lower algorithmic dimension $\lad F$ of $F$ is defined by
$\lad F
=
\sup_{x\in F}\liminf_{n \to \infty} H(\rest{x}{n})/n$.
\qed
\end{definition}
Thus,
for every $\alpha\in\R$,
\vspace*{-2mm}
\begin{equation}\label{lads}
  \lad\{\alpha\}=\liminf_{n \to \infty} \frac{H(\rest{\alpha}{n})}{n}.
\end{equation}

Independently of us,
Lutz \cite{Lutz00} introduced the notion of constructive dimension of
an individual real $\alpha$
using the notion of lower semicomputable $s$-supergale
with $s\in[0,\infty)$,
and then Mayordomo \cite{Mayor02} showed that,
for every real $\alpha$,
the constructive dimension of $\alpha$ equals to
the right-hand side of \eqref{lads}.
Thus, the constructive dimension of $\alpha$
is precisely the lower algorithmic dimension $\lad\{\alpha\}$ of
$\alpha$
for every real $\alpha$.

Using Lemma~\ref{sTr-iff-wtr} below,
we can convert each of
all the conditions
in Theorem~\ref{partial randomness}
into a characterization of
the lower algorithmic dimension $\lads{\alpha}$
for any r.e.~real $\alpha$.

\begin{lemma}\label{sTr-iff-wtr}
Let $\alpha\in\R$.
For every $t\in[0,\infty)$,
$\alpha$ is weakly Chaitin $t$-random
if $t<\lad\{\alpha\}$,
and
$\alpha$ is not weakly Chaitin $t$-random
if $t>\lad\{\alpha\}$.
\end{lemma}

\begin{proof}
Let $\alpha\in\R$,
and let $t\in[0,\infty)$.
Assume first that $t<\lad\{\alpha\}$.
Then, since $\lads{\alpha}n\le H(\rest{\alpha}{n})+o(n)$ for all $n\in\N^+$,
we see that
\begin{equation*}
  tn+\left(\lads{\alpha}-t-\frac{o(n)}{n}\right)n
  \le \lads{\alpha}n-o(n)\le H(\rest{\alpha}{n})
\end{equation*}
for all $n\in\N^+$.
Thus,
since $\lads{\alpha}-t-o(n)/n>0$ for all sufficiently large $n$,
we see that $\alpha$ is weakly Chaitin $t$-random.

On the other hand,
assume that $\alpha$ is weakly Chaitin $t$-random.
Then we see that
$t\le \liminf_{n \to \infty} H(\rest{\alpha}{n})/n=\lads{\alpha}$.
Thus,
if $t>\lad\{\alpha\}$ then $\alpha$ is not weakly Chaitin $t$-random.
This completes the proof.
\qed
\end{proof}

For example,
using Lemma~\ref{sTr-iff-wtr},
the condition (iii) in Theorem~\ref{partial randomness}
is converted as follows.
In this paper,
we interpret the supremum $\sup\emptyset$ of the empty set
as $0$.

\begin{theorem}\label{liminf-otl}
Let $\alpha$ be an r.e.~real.
Then,
for every $t\in(0,1]$,
$\alpha$ is $\Omega(t)$-like if $t<\lads{\alpha}$,
and
$\alpha$ is not $\Omega(t)$-like if $t>\lads{\alpha}$.
Thus,
\begin{equation*}
  \lads{\alpha}
  =
  \sup\{\,t\in(0,1]\mid
  \text{$\alpha$ is $\Omega(t)$-like}\,\}.
\end{equation*}
\vspace*{-12.0mm}\\
\qed
\end{theorem}

On the other hand,
the condition (viii) in Theorem~\ref{partial randomness}
is converted as follows,
using Lemma~\ref{sTr-iff-wtr}.
Here $\Rc$ denotes the set of all computable reals.

\begin{theorem}\label{liminf-upt}
Let $\alpha$ be an r.e.~real with $0<\alpha<1$,
Then,
for every $t\in(0,1]\cap\Rc$,
if $t<\lads{\alpha}$ then
$\alpha=\sum_{s\in\X}m(s)^{\frac{1}{t}}$
for some universal probability $m$,
and if $t>\lads{\alpha}$ then
$\alpha\neq\sum_{s\in\X}m(s)^{\frac{1}{t}}$
for any universal probability $m$.
Thus,
$\lad\{\alpha\}=\sup S$,
where $S$ is the set of all $t\in(0,1]\cap\Rc$ such that
$\alpha=\sum_{s\in\X}m(s)^{\frac{1}{t}}$
for some universal probability $m$.
\qed
\end{theorem}

In the same manner,
using Lemma~\ref{sTr-iff-wtr}
we can
convert each of the remaining eight conditions
in Theorem~\ref{partial randomness}
also
into a characterization of the lower algorithmic dimension
of an r.e.~real.
In a full version of this paper,
we will describe
the complete list of the ten characterizations
of the lower algorithmic dimension
obtained from Theorem~\ref{partial randomness}.

\section{Further properties of $T$-convergence}
\label{T-convergence}

In this section,
we investigate further properties of the notion of $T$-convergence.
First,
as one of the applications of Theorem~\ref{partial randomness},
the following theorem can be obtained.

\begin{theorem}\label{T-convergent->T-compressible}
Suppose that $T$ is computable.
For every r.e.~real $\alpha$,
if $\alpha$ is $T$-convergent,
then $\alpha$ is $T$-compressible. 
\end{theorem}

\begin{proof}
Using (vii) $\Rightarrow$ (iv) of Theorem~\ref{partial randomness},
we see that,
for every r.e.~$T$-convergent real $\alpha$,
$H(\rest{\alpha}{n})\le H(\rest{\Omega_U(T)}{n})+O(1)$
for all $n\in\N^+$.
It follows from Theorem~\ref{pomgd} (i) that
$\alpha$ is $T$-compressible
for every r.e.~$T$-convergent real $\alpha$.
\qed
\end{proof}

In the case of $T<1$,
the converse of Theorem \ref{T-convergent->T-compressible}
does not hold,
as seen in
Theorem~\ref{p2pt} below
in a sharper form.
Theorem~\ref{p2pt} can be proved
partly
using (vii) $\Rightarrow$ (ix) of Theorem~\ref{partial randomness}.

\begin{theorem}\label{p2pt}
Suppose that $T$ is computable and $T<1$.
Then there exists an r.e.~real $\eta$ such that
(i) $\eta$ is weakly Chaitin $T$-random and $T$-compressible,
and
(ii) $\eta$ is not $T$-convergent.
\qed
\end{theorem}

Let $T_1$ and $T_2$ be arbitrary computable reals
with $0<T_1<T_2<1$,
and let $V$ be an arbitrary optimal computer.
By Theorem~\ref{pomgd} (i) and Theorem~\ref{T-convergent->T-compressible},
we see that
the r.e.~real $\Omega_V(T_2)$ is not $T_1$-convergent
and therefore
every computable, increasing sequence $\{a_n\}$ of rationals
which converges to $\Omega_V(T_2)$ is not $T_1$-convergent.
At this point,
conversely,
the following question arises naturally:
Is there any computable, increasing sequence of rationals
which converges to $\Omega_V(T_1)$ and
which is not $T_2$-convergent~?
We can answer this question affirmatively
in the form of Theorem~\ref{expansion} below.

\begin{theorem}\label{expansion}
Let $T_1$ and $T_2$ be arbitrary computable reals
with $0<T_1<T_2<1$.
Then
there exist an optimal computer $V$ and
a computable, increasing sequence $\{a_n\}$ of
rationals such that
(i) $\Omega_V(T_1)=\lim_{n\to\infty}a_n$,
(ii) $\{a_n\}$ is $T$-convergent for every $T\in(T_2,\infty)$, and
(iii) $\{a_n\}$ is not $T$-convergent for every $T\in(0,T_2]$.
\qed
\end{theorem}

\vspace*{-7mm}

\section{Concluding remarks}
\label{conclusion}

In this paper,
we have generalized
the equivalent characterizations of randomness
of a recursively enumerable real
over the notion of partial randomness,
so that the generalized characterizations
are all equivalent to the weak Chaitin $T$-randomness.
As a stronger notion of partial randomness of a real $\alpha$,
Tadaki \cite{T99,T02} introduced
the notion of the Chaitin $T$-randomness of $\alpha$,
which is defined as the condition on $\alpha$ that
$\lim_{n\to\infty}H(\rest{\alpha}{n})-Tn=\infty$.%
\footnote{
The actual separation of the Chaitin $T$-randomness
from the weak Chaitin $T$-randomness is done by
Reimann and Stephan \cite{RS05}.}
Thus,
future work may aim at modifying
our equivalent characterizations of partial randomness
so that they
become
equivalent to the Chaitin $T$-randomness.

\subsubsection*{Acknowledgments.}

This work was supported
by KAKENHI, Grant-in-Aid for Scientific Research (C) (20540134),
by SCOPE
from the Ministry of Internal Affairs and Communications of Japan,
and by CREST from Japan Science and Technology Agency.



\end{document}